\documentclass[a4paper]{article}

\usepackage{algorithm}
\usepackage[noend]{algpseudocode}
\usepackage{amsmath,amssymb,amsthm}
\usepackage{authblk} % author names of a paper
\usepackage[american]{babel}
\usepackage{balance} % balance the two-column environment
\usepackage{fancyhdr} % custom header
\usepackage{enumitem} % customize enumerations 
\usepackage{float} % advanced placement options
\usepackage[T1]{fontenc}
\usepackage{fullpage} % doesn't work with fancyhdr
\usepackage{graphicx}
\usepackage[utf8]{inputenc}
\usepackage{layout}
\usepackage{mathtools} % \coloneqq, \DeclarePairedDelimiter
\usepackage{mdframed}
\usepackage{multirow} % table with split rows
\usepackage{tikz}
\usepackage{thmtools}
\usepackage{url}
\usepackage[normalem]{ulem} % uline, sout
\usepackage{xcolor}
\usepackage[backref=page, pagebackref=true, linktocpage, hypertexnames=false]{hyperref} % should be loaded here
\usepackage{cleveref} % cleveref after hyperref

\newlist{assertions}{enumerate}{10}
\setlist[assertions]{label={\arabic*}.,ref={\arabic*}}
\crefname{assertionsi}{assertion}{assertions}
\Crefname{assertionsi}{Assertion}{Assertions}

\usepackage{bm} % bold math, has to come after all the fonts have been loaded

\definecolor{redlink}{rgb}{0.6, 0, 0}
\definecolor{greenlink}{rgb}{0, 0.6, 0}
\definecolor{bluelink}{rgb}{0, 0, 0.6}
\hypersetup{
	colorlinks=true,
	linkcolor=redlink, % internal link
	urlcolor=redlink, % url
	citecolor=greenlink, % citation
}

\theoremstyle{definition}
\newtheorem{definition}{Definition}[section]

\theoremstyle{plain}
\newtheorem{theorem}{Theorem}[section]
\newtheorem{lemma}{Lemma}[section]
\newtheorem{corollary}{Corollary}[section]

\theoremstyle{remark}
\newtheorem{remark}{Remark}[section]
\newtheorem{claim}{Claim}[section]

\crefname{claim}{claim}{claims}

\newcommand{\eps}{\epsilon} % epsilon abbreviation
\newcommand{\ignore}[1]{} % ignoring parts of text 
\DeclarePairedDelimiter\ceil{\lceil}{\rceil}
\DeclarePairedDelimiter\floor{\lfloor}{\rfloor}

\title{Sparse Euclidean Spanners with Tiny Diameter: \\A Tight Lower Bound}
\author[1]{Hung Le}
\affil[1]{University of Massachusetts Amherst}
\author[2]{Lazar Milenković}
\author[3]{Shay Solomon}
\affil[2,3]{Tel Aviv University}

\begin{document}
\maketitle
\begin{abstract}
In STOC'95 \cite{ADMSS95} Arya et al. showed that any set of $n$ points in $\mathbb R^d$ admits a $(1+\epsilon)$-spanner with hop-diameter at most 2 (respectively, 3) and $O(n \log n)$ edges (resp., $O(n \log \log n)$ edges). They also gave a general upper bound tradeoff of hop-diameter at most $k$ and $O(n \alpha_k(n))$ edges, for any $k \ge 2$.
The function $\alpha_k$ is the inverse of a certain Ackermann-style function at the $\lfloor k/2 \rfloor$th level of the primitive recursive hierarchy, where $\alpha_0(n) = \lceil n/2 \rceil$, $\alpha_1(n) = \left\lceil \sqrt{n} \right\rceil$, $\alpha_2(n) = \lceil \log{n} \rceil$, $\alpha_3(n) = \lceil \log\log{n} \rceil$, $\alpha_4(n) = \log^* n$, $\alpha_5(n) = \lfloor \frac{1}{2} \log^*n \rfloor$, \ldots. Roughly speaking, for $k \ge 2$ the function $\alpha_{k}$ is close to $\lfloor \frac{k-2}{2} \rfloor$-iterated log-star function, i.e., $\log$ with $\lfloor \frac{k-2}{2} \rfloor$ stars. Also, $\alpha_{2\alpha(n)+4}(n) \le 4$, where $\alpha(n)$ is the one-parameter inverse Ackermann function, which is an extremely slowly growing function.

Whether or not this tradeoff is tight has remained open, even for the cases $k = 2$ and $k = 3$. Two lower bounds are known: The first applies only to spanners with stretch 1 and the second is sub-optimal and applies only to sufficiently large (constant) values of $k$. In this paper we prove a tight lower bound for any constant $k$: For any fixed 
$\epsilon > 0$, any $(1+\epsilon)$-spanner for the uniform line metric with hop-diameter at most $k$ must have at least $\Omega(n \alpha_k(n))$ edges.
\end{abstract}
\section{Introduction}
Consider a set $S$ of $n$ points in $\mathbb R^d$ and a real number $t \ge 1$.
A weighted graph $G = (S,E,w)$ in which the weight function is given by the Euclidean distance, 
i.e.,  $w(x,y) =  \|x-y\|$ for each $e =(x,y)\in E$, 
is called a \emph{geometric graph}.
We say that a geometric graph $G$ is a \emph{$t$-spanner} for $S$ if for every pair $p,q \in S$ of distinct points,
there is a path in $G$ between $p$ and $q$ whose {\em weight}
(i.e., the sum of all edge weights in it) 
is at most $t$ times the
Euclidean distance $\|p-q\|$ between $p$ and $q$. 
Such a path 
is called a \emph{$t$-spanner path}. The problem of constructing
Euclidean spanners 
has been studied intensively
over the years \cite{Chew86,KG92,ADDJS93,CK93,DN94,ADMSS95,DNS95,AS97,RS98,AWY05,CG06,DES09tech,SE10,Sol14,ES15,LS19}. 
Euclidean spanners are of importance both in theory and in practice,
as they enable approximation of the complete Euclidean
graph in a more succinct form; in particular, they 
find a plethora of applications, e.g., in geometric approximation algorithms, network topology design, geometric distance oracles, distributed systems, design of parallel machines, and other areas \cite{DN94,LNS98,RS98,GLNS02,GNS05,GLNS08,HP00,MP00}.
We refer the reader to the book by Narasimhan and Smid \cite{NS07}, which provides a thorough account on Euclidean spanners and their applications.

In terms of applications, the most basic requirement from a spanner (besides achieving a small stretch) is to be \emph{sparse},  i.e.,
to have only a small number of edges. 
However, for many applications, the spanner is required to
preserve some additional properties of the underlying complete graph.
One such property, which plays a key role in various applications (such as to routing protocols) \cite{AMS94,AM04,AWY05,CG06,DES09tech,KLMS21},  is the \emph{hop-diameter}:
a $t$-spanner for $S$ is said to have an hop-diameter of $k$ if, for any $p,q \in S$, there is a $t$-spanner path between $p$ and $q$ with at most $k$ edges (or hops). 

\subsection{Known upper bounds}
\paragraph*{1-spanners for tree metrics.}
We denote the tree metric induced by an $n$-vertex (possibly weighted) rooted tree $(T,rt)$ by $M_T$. 
A spanning subgraph $G$ of $M_T$ is said to be a \emph{1-spanner} for $T$, if for every pair of vertices, their distance in $G$ is equal to their distance in $T$.
The problem of constructing 1-spanners for tree metrics is a fundamental one, and has  been studied quite extensively over the years, also in more
general settings, such as planar
metrics \cite{Thor95}, general metrics \cite{Thor92} and general graphs \cite{BGJRW09}. 
This problem is also intimately  related to the extremely well-studied problems of
computing partial-sums and online product queries in semigroup and their variants (see \cite{Tarj79,Yao82,AS87,CR91,PD04,AWY05}, and the references therein).

Alon and Schieber \cite{AS87} and Bodlaender et al.\ \cite{BTS94} showed that 
for any $n$-point tree metric, a 1-spanner with diameter 2 (respectively, 3) and $O(n \log n)$ edges (resp., $O(n \log \log n)$ edges)
can be built within time linear in its size.  
For $k \ge 4$, Alon and Schieber \cite{AS87} showed that 1-spanners with diameter at most $2k$ 
and $O(n \alpha_k(n))$ edges can be built in $O(n \alpha_k(n))$ time. 
The function $\alpha_k$ is the inverse of a certain Ackermann-style function at the $\lfloor k/2 \rfloor$th level of the primitive recursive hierarchy, where $\alpha_0(n) = \ceil*{n/2}$, $\alpha_1(n) = \ceil*{\sqrt{n}}$, $\alpha_2(n) = \ceil{\log{n}}$, $\alpha_3(n) = \ceil{\log\log{n}}$, $\alpha_4(n) = \log^* n$, $\alpha_5(n) = \floor*{\frac{1}{2} \log^*n}$, etc. Roughly speaking, for $k \ge 2$ the function $\alpha_{k}$ is close to $\lfloor \frac{k-2}{2} \rfloor$-iterated log-star function, i.e., $\log$ with $\lfloor \frac{k-2}{2} \rfloor$ stars. Also, $\alpha_{2\alpha(n)+2}(n) \le 4$, where $\alpha(n)$ is the one-parameter inverse Ackermann function, which is an extremely slowly growing function. (These functions are formally defined in \Cref{sec:ackermann}.)
Bodlaender et al.\ \cite{BTS94} constructed 1-spanners 
with diameter at most $k$ and $O(n \alpha_k(n))$ edges, with a high running time.
Solomon~\cite{Sol13} gave a construction that achieved the best of both worlds: a tradeoff of $k$ versus $O(n\alpha_k(n))$ between the hop-diameter and the number of edges in linear time of $O(n\alpha_k(n))$.

Alternative constructions, given by Yao~\cite{Yao82} for line metrics and later extended by Chazelle~\cite{Chaz87} to general tree metrics, achieve a tradeoff of $m$ edges versus $\Theta( \alpha(m,n))$ hop-diameter, where $\alpha(m,n)$ is the two-parameter inverse Ackermann function (see \Cref{def:invAck}). However, these constructions provide 1-spanners
with diameter $\Gamma' \cdot k$ rather than $2k$ or $k$, for some constant $\Gamma' > 30$.

\paragraph*{\bm{$(1+\eps)$}-spanners.~}
In STOC'95 Arya~et~al.~\cite{ADMSS95} proved the so-called ``Dumbbell Theorem'', which states that,
for any $d$-dimensional Euclidean space, a $(1+\eps,O(\frac{\log(1/\eps)}{\eps^d}))$-tree cover %construction 
can be constructed in $O(\frac{\log(1/\eps)}{\eps^d}\cdot n\log{n}+\frac{1}{\eps^{2d}}\cdot n)$ time; see \cref{sec:prelim} for the definition of tree cover. The Dummbell Theorem implies that any construction of 1-spanners for tree metrics can be translated into a construction of Euclidean $(1+\eps)$-spanners. Applying the construction of 1-spanners for tree metrics from \cite{Sol13}, this gives rise to an optimal $O(n \log n)$-time construction (in the algebraic computation tree (ACT) model\footnote{Refer to Chapter 3 in \cite{NS07} for the definition of the ACT model. A matching lower bound of $\Omega(n \log n)$ on the time needed to construct Euclidean spanners is given in \cite{CDS01}.}) of Euclidean $(1+\eps)$-spanners.
This result  can be generalized (albeit not in the ACT model) for the wider family of {\em doubling metrics}, by using the tree cover theorem of Bartal~et~al.~\cite{BFN19}, which generalizes the Dumbbell Theorem of \cite{ADMSS95} for arbitrary doubling metrics.

\subsection{Known lower bounds}
The first lower bound on 1-spanners for tree metrics was given by Yao~\cite{Yao82} and it establishes a tradeoff of $m$ edges versus hop-diameter of $\Omega(\alpha(m,n))$ for the uniform line metric.
Alon~and~Schieber~\cite{AS87} gave a stronger lower bound on 1-spanners for the uniform line metric: hop-diameter $k$ versus $\Omega(n\alpha_k(n))$ edges, for any $k$; it is easily shown that this lower bound implies that of \cite{Yao82} 
(see \Cref{lemma:tradeoff}),
 but the converse is not true.

The above lower bounds apply to 1-spanners. There is also a lower bound on $(1+\eps)$-spanners that applies to line metrics, by Chan and Gupta \cite{CG06}, 
which extends that of \cite{Yao82}: $m$ edges versus hop-diameter of $\Omega(\alpha(m,n))$. 
As mentioned already concerning this tradeoff, it only provides a meaningful bound for sufficiently large values of hop-diameter (above say 30), and it does not apply to hop-diameter values that approach 1, which is the focus of this work.
More specifically, it can be used to show that any $(1+\eps)$-spanner for a certain line metric with hop-diameter at most $k$ must have $\Omega(n \alpha_{2k+6} (n))$ edges. When $k=2$ (resp. $k=3$), this gives $\Omega(n\log^{****}{n})$ (resp. $\Omega(n\log^{*****}{n})$) edges,
which is   far from the upper bound of $O(n \log n)$ (resp., $O(n \log\log n)$). 
Furthermore, the line metric used in the  proof of \cite{CG06} is not as basic as the uniform line metric --- it is derived from hierarchically well-separated trees (HSTs), and to achieve the result for line metrics, an embedding from HSTs to the line with an appropriate separation parameter is employed. The resulting line metric is very far from a uniform one and its aspect ratio\footnote{The {\em aspect ratio} of a metric is the ratio of the maximum pairwise distance to the minimum one.} depends on the stretch --- it will be super-polynomial whenever $\eps$ is sufficiently small or sufficiently large; 
of course, the aspect ratio of the uniform line metric (which is the metric used by \cite{Yao82,AS87}) is linear in $n$. 
As point sets arising in real-life applications (e.g., for various random distributions) have polynomially
bounded aspect ratio, it is natural to ask whether one can achieve a lower bound for a point set of polynomial aspect ratio.

\subsection{Our contribution}
We prove that any $(1+\eps)$-spanner for the uniform line metric with hop-diameter $k$ must have at least $\Omega(n\alpha_k(n))$ edges, for any constant $k \ge 2$.

\begin{theorem}\label{thm:main}
For any positive integer $n$, any integer $k\ge2$ and any $\eps \in [0, 1/2]$, any $(1+\eps)$-spanner  with hop-diameter $k$ for the uniform line metric with $n$ points must contain at least $\Omega(\frac{n}{2^{6\floor{k/2}}} \alpha_{k}(n))$ edges.
\end{theorem}
Interestingly, our lower bound applies also to any $\eps > 1/2$, where the bound on the number of edges reduces linearly with $\eps$, i.e., it becomes $\Omega(n\alpha_k(n) /\eps)$.
We stress that our lower bound instance, namely the uniform line metric, does not depend on $\eps$, and the lower bound that it provides holds {\em simultaneously for all values of $\eps$}.

Although our lower bound on the number of edges coincides with $\Omega(n\alpha_k(n))$ only for constant  $k$, we note that the values of $k$ of interest range between $1$ and $O(\alpha(n))$,
where $\alpha(\cdot)$ is a very slowly growing function, e.g., $\alpha(n)$ is asymptotically much smaller than $\log^*n$; we formally define $\alpha(n)$ in \Cref{sec:ackermann}.
Indeed, as mentioned, for $k=2\alpha(n)+4$, we have $\alpha_{2\alpha(n) + 4}(n) \le 4$, 
and clearly any spanner must have $\Omega(n)$ edges.
Thus the gap between our lower bound on the number of edges and $\Omega(n\alpha_k(n))$,
namely, a multiplicative factor of $2^{6\floor{k/2}}$, which in particular is no greater than $2^{O(\alpha(n)}$, is very small.

For technical reasons we prove a more general lower bound, stated in \Cref{thm:k}. In particular, we need to consider a more general notion of Steiner spanners\footnote{A Steiner spanner for a point set $S$ is a spanner that may contain additional Steiner points (which do not belong to $S$). Clearly, a lower bound for Steiner spanners also applies to ordinary spanners.}, and to prove the lower bound for a certain family of line metrics to which the uniform line metric belongs;
\Cref{thm:main} follows directly from \Cref{thm:k}. See \Cref{sec:prelim} for the definitions.

For constant values of k, \Cref{thm:main} strengthens the lower bound shown by \cite{AS87}, which applies only to stretch 1, whereas our tradeoff holds for arbitrary stretch. Whether or not the term $\frac{1}{2^{6\floor{k/2}}}$ in the bound on the number of edges in \Cref{thm:main} can be removed is left open by our work. As mentioned before, we show in \Cref{sec:tradeoff} that this tradeoff implies the tradeoff by \cite{Yao82} (for stretch 1) and \cite{CG06} (for larger stretch).

\subsection{Proof Overview}

The starting point of our lower bound argument is the one by \cite{AS87} for stretch 1, which applies to the uniform line metric, $U(n)$. The argument of \cite{AS87} crucially relies on the following {\em separation} property: for any 4 points $i,j,i',j'$ in $U(n)$ such that $i \le j < i' \le j'$, any 1-spanner path between $i$ and $j$ is completely disjoint to any 1-spanner path between $i'$ and $j'$. Thus, for any two subintervals that do not overlap --- and are thus separated in the most basic sense, any 1-spanner path between a pair of points in one of the sub-intervals does not overlap any 1-spanner path in the other sub-interval. With this separation property in mind, a rather straightforward inductive argument can be employed. Consider for concreteness the easiest case $k = 2$. The number of spanner edges satisfies the recurrence $T_2(n) = 2T_2(n/2) + \Omega(n)$. Indeed, by induction, the number of spanner edges for the left $n/2$ points is at least $T_2(n/2)$, and the same goes for the right $n/2$ points, and importantly the corresponding sets of edges are disjoint by the separation property. All is left is to reason that there are $\Omega(n)$ \emph{cross edges}, which are the edges whose one endpoint belongs to the left $n/2$ points and the other to the right $n/2$ points; the separation property implies that the cross edges are disjoint to the aforementioned edge sets. Alas, when the stretch grows from 1 to $1+\eps$, even for a tiny $\eps > 0$, the basic separation property no longer holds, so it is possible to have significant overlaps between the edge sets corresponding to the recursive calls of the left half and right half of the points, and between them and the cross edges. Of course, as $k$ increases, such issues regarding overlaps become even more challenging. To overcome these issues, \cite{CG06} resorted to a different metric space, which is far from a uniform line metric to an extent that renders such issues of overlapping negligible. As mentioned before, the resulting metric of \cite{CG06} has a super-polynomial aspect ratio whenever $\eps$ is sufficiently small or sufficiently large. Our argument, on the other hand, applies to the uniform line metric; coping with overlaps issues in this case is nontrivial. Next, we highlight some of the challenges that our argument overcomes.

The easiest case is $k=2$.
We view the line metric $U(1,n)$ as an interval $[1,n]$ of all integers from $1$ to $n$, and also consider sub-intervals $[i,j]$ of all integers from $i$ to $j$, for all $1 \le i<j \le n$.
Our first observation is that one can achieve a similar separation property as in the case of stretch 1 by, roughly speaking, restricting the attention to a subset of $\Omega(n)$ points that are sufficiently far from the boundaries of sub-intervals. More concretely, for $\eps = 1/2$, consider the $n/2$ points closest to $n/2$ (i.e., the sub-interval $[n/4, 3n/4]$):  While these $n/2$ points still induce $\Omega(n)$ cross edges, note that any $(1+\eps)$-spanner path between any pair of such points is contained in $[1,n]$. To achieve the required separation property, we would like to apply this observation inductively. 
However, a naive application is doomed to failure.
Indeed, if we used the induction hypothesis only on the points inside $[n/4,3n/4]$, i.e., on the points of $[n/4, n/2]$ as the left interval and the points of $[n/2+1,3n/4]$
as the right interval, the number of cross edges would degrade by a factor of 2. Losing a factor of 2 at each recursion level cannot give a lower bound larger than $\Omega(n)$. Instead,
we apply the induction hypothesis on all $n/2$ points of the left half and on all $n/2$ points of the right half, but restrict the attention to the $n/4$ points that are closest to $n/4$ in the left half and to the $n/4$ points that are closest to $3n/4$ in the right half, which gives us $n/2$ points in total, just as with the first recursion level. In this way at each level of recursion we restrict the attention to a different set of points, but of the same size $n/2$, and in this way avoid the loss.
The case $k=3$ is handled similarly, but is more intricate --- primarily since the basic argument of \cite{AS87} (for stretch 1) for $k = 3$ is more intricate than for $k = 2$. 
We refer the reader to the proof of \Cref{lemma:lb2} (resp. \Cref{lemma:lb3}) for the  proof of the case $k=2$ (resp. $k=3$).

For $k\ge 4$, the argument is considerably more involved. 
As with \cite{AS87}, we divide the interval $[1,n]$ into consecutive sub-intervals of size $\alpha_{k-2}(n)$ --- the goal is to show that the number of spanner edges satisfies the recurrence $T_k(n) = (n/\alpha_{k-2}(n)) \cdot T_k(\alpha_{k-2}(n)) + \Omega(n)$, implying that $T_k(n) = \Omega(n\alpha_k(n))$. Using the separation property as in the case $k=2$, employing the induction hypothesis on each of the $n/\alpha_{k-2}(n)$ subintervals yields the term $(n/\alpha_{k-2}(n)) \cdot T_k(\alpha_{k-2}(n))$. The crux is to show that there are $\Omega(n)$ cross edges, now defined as edges with endpoints in different sub-intervals. 
To ensure that the separation property holds, as in the case $k = 2$, we consider only the $n/2$ points of $[n/4, 3n/4]$. We distinguish between \emph{global} and \emph{non-global} points: A point is called global if it is incident on at least one cross edge, and non-global otherwise. If
a constant fraction of points in $[n/4,3n/4]$ are global, 
we must have $\Omega(n)$ cross edges by definition. The complementary case, where a constant fraction of points in $[n/4,3n/4]$ are non-global, is the interesting one.
We'd like to use the induction hypothesis for $k-2$, as in \cite{AS87}, to reason that the number of cross edges induced by the non-global points is at least $T_{k-2}(n/\alpha_{k-2}(n)) = \Omega(n)$. The problem is that the non-global points induce a line metric that is {\em not  uniform}, hence we cannot apply the induction hypothesis --- overcoming this obstacle is the key difficulty in our argument. 
To be able to apply the induction hypothesis, we prove all the results on a generalized line metric. Consider first what we call a {\em $t$-sparse line metric}, which contains a single point in every consecutive sub-interval of size $t$. 
Such a line metric would suffice for applying the induction hypothesis (with hop-diameter $k-2$), {\em if all sub-intervals inside $[n/4,3n/4]$ contained at least one non-global point}. In this special case, we could apply the induction hypothesis on an $\alpha_{k-2}(n)$-sparse line metric, and if we're always in this special case, we'll always have a $t$-sparse line metric for a growing parameter $t$, and it is not difficult to show that the induction step will work fine. Alas, we only know that there is a constant fraction of points around the middle of the interval and the intervals containing them form a subspace of a $t$-sparse line metric, which is possibly very different than this special case; thus, in general, we are unable to apply the induction hypothesis in such a way. 
To overcome this hurdle, we prove a stronger, more general lower bound, which concerns subspaces of $t$-sparse line metrics, where a constant fraction of the points is missing. On the bright side, such generalized spaces 
provide the required flexibility for carrying out the induction step.
On the negative side, each invocation of the induction hypothesis with a smaller value of $k$ incurs some multiplicative loss in the number of considered points, yielding the exponential on $k$ slack in our lower bound. 
We refer the reader to the proof of \Cref{thm:k}.

\section{Preliminaries}\label{sec:prelim}
\begin{definition}[Tree covers]
Let $M_X = (X, \delta_X)$ be an arbitrary metric space.
We say that a weighted tree $T$ is a {\em dominating} tree for $M_X$ if $X \subseteq V(T)$  and it holds that $\delta_T(x,y) \ge \delta_X(x,y)$, for every $x,y\in X$.
For $\gamma \ge 1$ and an integer $\zeta\ge1$, a {\em $(\gamma, \zeta)$-tree cover} of  $M_X = (X,\delta_X)$ is a collection of $\zeta$ {\em dominating trees} for $M_X$, such that for every $x,y \in X$, there exists a tree $T$ with $d_T(u,v) \le \gamma \cdot \delta_X(u,v)$;
we say that the {\em stretch} between $x$ and $y$ in $T$ is at most $\gamma$,
and the parameter $\gamma$ is referred to as the {\em stretch} of the tree cover. 
\end{definition}

\begin{definition}[Uniform line metric]
A uniform line metric $U=(\mathbb{Z}, d)$ is a metric on a set of integer points such that the distance between two points $a,b \in \mathbb{Z}$, denoted by $d(a,b)$ is their Euclidean distance, which is $|a-b|$.
For two integers $l,r\in\mathbb{Z}$, such that $l \le r$, we define a uniform line metric on an interval $[l, r]$, denoted by $U(l, r)$, as a subspace of $U$ consisting of all the integer points $k$, such that $l \le k \le r$. We use $U(n)$ to denote a uniform line metric on the interval $[1,n]$.
\end{definition}

Although we aim to prove the lower bound for uniform line metric, the inductive nature of our argument requires several generalizations of the considered metric space and spanner. 

\begin{definition}[$t$-sparse line metric]\label{def:tsparse}
Let $l$ and $r$ be two integers such that $l<r$. We call metric space $U((l,r),t)$ $t$-sparse if:
\begin{itemize}
\item It is a subspace of $U(l,r)$. 
\item Each of the consecutive intervals of $[l,r]$ of size $t$ ($[l, l+t], [l+t+1, l+2t], \dots$) contains exactly one point. These intervals are called $((l,r),t)$-intervals and the point inside each such interval is called \emph{representative} of the interval.
\end{itemize}
\end{definition}

\emph{Throughout the paper, we will always consider Steiner spanners that can contain arbitrary points from the uniform line metric.}

\begin{definition}[Global hop-diameter]\label{def:spanner} For any two integers $l,r$ such that $r=l+nt-1$, let $U((l,r),t)$ be a $t$-sparse line metric with $n$ points and let $X$ be a subspace of $U((l,r),t)$. An edge that connects two points in $U((l,r),t)$ is \emph{$((l,r),t)$-global} if it has endpoints in two different $((l,r),t)$-intervals of $U((l,r),t)$. 
A spanner on $X$ with stretch $(1+\eps)$ has its \emph{$((l,r),t)$-global hop-diameter} bounded by $k$ if every pair of points in $X$ has a path of stretch at most $(1+\eps)$ consisting of at most $k$ $((l,r),t)$-global edges.
\end{definition}

For ease of presentation, we focus on $\eps \in [0, 1/2]$, as this is the basic regime.
Our argument naturally extends to any $\eps>1/2$, with the lower bound degrading by a factor of $1/\eps$.

\begin{lemma}[Separation property]\label{lemma:sep}
Let $l, r, t \in \mathbb{N}$, $l\le r$, $t\ge 1$ and let $i \coloneqq \ceil{\frac{1+\eps/2}{1+\eps}l + \frac{\eps/2}{1+\eps}r}$, and $j \coloneqq \floor{\frac{\eps/2}{1+\eps}l + \frac{1+\eps/2}{1+\eps}r}$.
Let $a, b$ be two points in $U((l,r),t)$ such that $i\le a < b \le j$. Then, any $(1+\eps)$-spanner path between $a$ and $b$ contains points strictly inside $[l,r]$.
\end{lemma}
\begin{proof}
Consider a spanner path between $a$ and $b$ which contains a point $q$ outside $U((l,r),t)$ such that $q > r$. (A similar argument holds for $q < l$.)
The length of any such path is at least $(b-a) + 2(q-b)$. Since $b \le j < r < q$, it holds that $2(q-b) > 2(r-j) =\frac{\eps}{1+\eps}(r-l)$.
On the other hand, since $i \le a < b \le j$, the distance between $a$ and $b$ is at most $b-a \le j-i = \frac{1}{1+\eps}(r-l)$.
The last two inequalities imply that $2(q-b) > \eps(b-a)$. It follows that the spanner path between $a$ and $b$ containing $q$ is of length greater than $(1+\eps)(b-a)$, i.e., it has a stretch bigger than $(1+\eps)$.
\end{proof}

\begin{corollary}\label{cor:sep}
For every integer $N\ge 34$ and any $t$-sparse line metric $U((1,N),t)$, any spanner path with stretch at most $3/2$ between metric points $a$ and $b$ such that $\floor{N/4} \le a \le b \le \ceil{3N/4}$ contains points strictly inside $[1,N]$.
\end{corollary}

\subsection{Ackermann functions}\label{sec:ackermann}
Following standard notions \cite{Tarjan75,AS87,Chaz87,NS07,Sol13}, we will introduce two very rapidly growing functions $A(k, n)$ and $B(k, n)$, which are variants of Ackermann's function. Later, we also introduce several inverses and state their properties that will be used throughout the paper.
\begin{definition}
For all $k \ge 0$, the functions $A(k,n)$ and $B(k,n)$ are defined as follows:
\begin{align*}
A(0, n) &\coloneqq 2n, \text{\emph{ for all }} n \ge 0,\\
A(k, n) &\coloneqq 
\begin{cases}
1 &\text{\emph{ if }} k \ge 1 \text{\emph{ and }} n = 0\\
A(k-1, A(k, n-1)) & \text{\emph{ if }} k \ge 1 \text{\emph{ and }} n \ge 1\\
\end{cases}\\
B(0, n) &\coloneqq n^2, \text{\emph{ for all }} n \ge 0,\\
B(k, n) &\coloneqq 
\begin{cases}
2 &\text{\emph{ if }} k \ge 1 \text{\emph{ and }} n = 0\\
B(k-1, B(k, n-1)) & \text{\emph{ if }} k \ge 1 \text{\emph{ and }} n \ge 1\\
\end{cases}
\end{align*}
\end{definition}
We now define the functional inverses of $A(k,n)$ and $B(k,n)$.
\begin{definition}%\label{def:RowInvAckermann}
For all $k \ge 0$, the function $\alpha_k(n)$ is defined as follows:
\begin{align*}
 \alpha_{2k}(n) &\coloneqq \min\{s \ge 0: A(k, s) \ge n\},\text{\emph{ for all }} n\ge 0\text{\emph{, and}}\\
	 \alpha_{2k+1}(n) &\coloneqq \min\{s \ge 0: B(k, s) \ge n\},\text{\emph{ for all }} n\ge 0,\text{{for all }} n\ge 0.
\end{align*}
\end{definition}
For technical convenience we define $\log{x} = 0$ for any $x \le 0$, $x \in \mathbb{R}$. All the logarithms are with base 2.
It is not hard to verify that $\alpha_0(n) = \lceil n/2\rceil$,
$\alpha_1(n) = \lceil \sqrt{n} \rceil$,
$\alpha_2(n)= \lceil\log{n}\rceil$,
$\alpha_3(n)= \lceil\log\log{n}\rceil$,
$\alpha_4(n)= \log^*{n}$,
$\alpha_5(n)= \lfloor \frac{1}{2}\log^*{n} \rfloor$, etc. 
We will use the following property of $\alpha_k(n)$.
\begin{lemma}[cf. Lemma 12.1.16. in \cite{NS07}]\label{lemma:alphakstep}
For each $k\ge 1$, we have:\\
$\alpha_{2k}(n) = 1 + \alpha_{2k}(\alpha_{2k-2}(n))$, for all $n\ge 2$, and\\
$\alpha_{2k+1}(n) = 1 + \alpha_{2k+1}(\alpha_{2k-1}(n))$, for all $n\ge 3$.
\end{lemma}

Finally, another functional inverse of $A(\cdot, \cdot)$ is defined as in \cite{Tarjan75,Chaz87,CG06,NS07,Sol13}.
\begin{definition} [Inverse Ackermann function]\label{def:invAck}
$\alpha\left(m,n\right) = \min \left\{ i | i \ge 1, A \left( i, 4\ceil{m/n} \right) > \log_2{n} \right\}$.
\end{definition}

Finally, for all $n\ge 0$, we introduce the Ackermann function as $A(n) \coloneqq A(n,n)$, and its inverse as $\alpha(n) = \min\{s \ge 0 : A(s) \ge n\}$. In \cite{NS07}, it was shown that  $\alpha_{2\alpha(n)+2}(n)\le 4$. We observe that $\alpha(n)$ satisfies $\alpha(n) \le \log^*n$ for any $n \ge 2$.

\section{Warm-up: lower bounds for hop-diameters 2 and 3}

In this section, we prove the lower bound for cases $k=2$ (\Cref{lemma:lb2} in \Cref{sec:lb2}) and $k=3$ (\Cref{lemma:lb3} in \Cref{sec:lb3}). In fact, we prove more general statements (\Cref{thm:2,thm:3}), which apply not only to uniform line metric, but to subspaces of $t$-sparse line metrics, where a constant fraction of the points is missing.
We use these general statements in \Cref{sec:lbk}, to prove the result for general $k$ (cf.\ \Cref{thm:k}).

\subsection{Hop diameter 2}\label{sec:lb2}

\begin{theorem}\label{thm:2}
For any two positive integers $n\ge 1000$ and $t$, and any two integers $l,r$ such that $r = l + nt-1$, let $U((l,r),t)$ be a $t$-sparse line metric with $n$ points and let $X$ be a subspace of $U((l,r),t)$ which contains at least $\frac{31}{32}n$ points. Then, for any choice of $\eps \in [0,1/2]$, any spanner on $X$ with $((l,r),t)$-global hop-diameter 2 and stretch $1+\eps$ contains at least $T'_2(n) \ge \frac{n}{256}\cdot\alpha_{2}(n)$ $((l,r),t)$-global edges which have both endpoints inside $[l,r]$.
\end{theorem}
\begin{remark}
Recall that we consider Steiner spanners, which could possibly contain additional Steiner points from the uniform line metric. 
\end{remark}
\begin{remark}
\Cref{thm:2} can be extended to $\eps > 1/2$. The only required change in the proof is to decrease the lengths of intervals by a factor of $1+\eps$, as provided by \Cref{lemma:sep}; it is readily verified that, as a result, the lower bound decreases by a factor of $\Theta(\eps)$.
\end{remark}

The theorem is proved in three steps. First, we prove \Cref{lemma:lb2}, which concerns uniform line metrics. Then, we prove \Cref{lemma:2prim} for a subspace that contains at least 31/32 fraction of the points of the original metric. In the third step, we observe that the same argument applies for $t$-sparse line metrics.

\begin{lemma}\label{lemma:lb2}
For any positive integer $n$, and any two integers $l,r$ such that $r = l + n-1$, let $U(l,r)$ be a uniform line metric with $n$ points. Then, for any choice of $\eps \in [0,1/2]$, any   spanner on $U(l,r)$ with hop-diameter 2 and stretch $1+\eps$ contains at least $T_2(n) \ge \frac{1}{16}\cdot n\log{n}$ edges which have both endpoints inside $[l,r]$.
\end{lemma}
\begin{proof}
Suppose without loss of generality that we are working on the uniform line metric $U(1,n)$. Let $H$ be an arbitrary   $(1+\eps)$-spanner for $U(1,n)$ with hop-diameter 2.

For the base case, we take  $n < 128$. In that case our lower bound is 
$\frac{n}{16}\cdot \log{n} < n-1$, which is a trivial lower bound for the number of edges in $H$.

For the proof of the inductive step, we can assume that $n \ge 128$. We would like to prove that the number of spanner edges in $H$ is lower bounded by $T_2(n)$, which satisfies recurrence $T_2(n) = 2T_2(\floor{n/2})+ \Omega(n)$ with the base case $T_2(n) = (n/16)\log{n}$ when $n\le 128$. Split the interval into two disjoint parts: the left part $[1,\floor{n/2}]$ and the right part $[\floor{n/2}+1, n]$. From the induction hypothesis on the uniform line metric $U(1,\floor{n/2})$ we know that any   spanner with hop-diameter 2 and stretch $1+\eps$ contains at least $T_2(\floor{n/2})$ edges that have both endpoints inside $[1,\floor{n/2}]$. Similarly, any   spanner for $U(\floor{n/2}+1,n)$ contains at least $T_2(\floor{n/2})$ edges that have both endpoints inside $[\floor{n/2}+1,n]$. This means that the sets of edges considered on the left side and the right side are disjoint. We will show below that there are $\Omega(n)$ edges that have one point on the left and the other on the right.

Consider the set $L$, consisting of the points inside $[n/4, \floor{n/2}]$ and the set $R$, consisting of the points in $[\floor{n/2}+1, 3n/4]$. From \Cref{cor:sep}, since $n$ is sufficiently large, we know that any $(1+\eps)$-spanner path connecting point $a \in L$ and $b\in R$ has to have all its points inside $[1,n]$. 
We use term \emph{cross edge} to denote any edge that has one endpoint in the left part and the other endpoint in the right part. 
We claim that any spanner with hop-diameter at most $2$ and stretch $1+\eps$ has to contain at least $\min(|L|,|R|)$ cross edges. Without loss of generality, assume that $|L| \le |R|$. Suppose for contradiction that the spanner contains less than $|L|$ cross edges. This means that at least one point in $x\in L$ is not connected via a direct edge to any point on the right. Observe that, for every point $r \in R$, the 2-hop spanner path between $x$ and $r$ must be of the form $(x,l_r,r)$ for some point $l_r$ in the left set. It follows that every $r\in R$ induces a different cross edge $(l_r, r)$. Thus, the number of cross edges, denoted by $|E_C|$, is $|R| \geq |L|$, which is a contradiction. From the definition of $L$ and $R$, we know that $\min(|L|,|R|) \ge n/4-2$, implying that the number of cross edges is at least $n/4-2 \ge 11n/64$, for all $n \ge 26$. (See also \Cref{fig:lb2} for an illustration.) Thus, we have:

\begin{align*}
T_2(n) ~=~ 2T_2(\floor{n/2}) + \frac{11n}{64} 
~\stackrel{\mbox{\tiny{induction}}}{\ge}~ 2\cdot\frac{\floor{n/2}}{16}\log\floor{n/2} +  \frac{11n}{64}
~\ge~ \frac{n}{16}\cdot \log{n}~,\\
\end{align*}
as claimed. 
\end{proof}

\begin{figure}[!bht]
\centering
\input{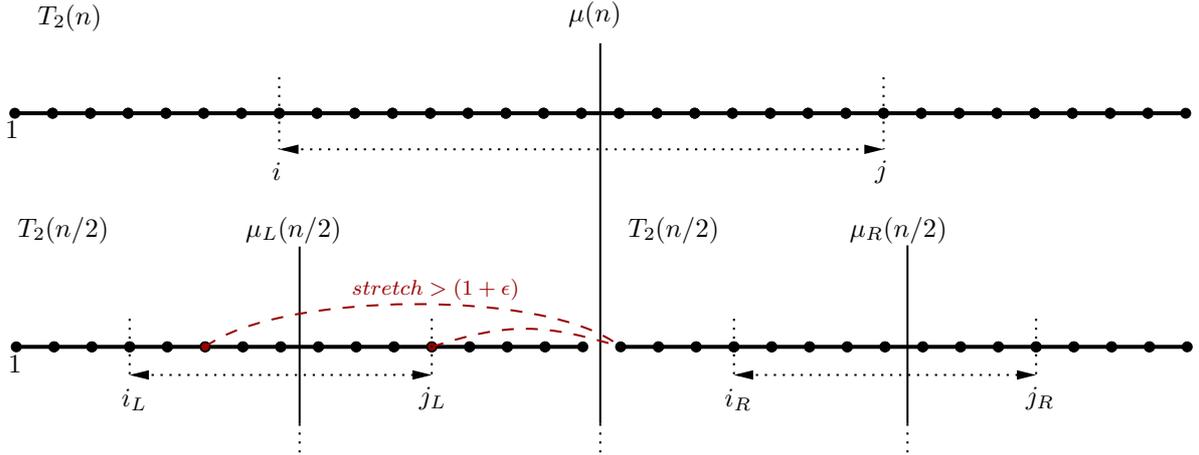}
\caption{An illustration of the first two levels of the recurrence for the lower bound for $k=2$ and $\eps=1/2$. We split the interval $U(1, n)$ into two disjoint parts.
In \Cref{lemma:lb2}, we show that there will be at least $\Omega(n)$ cross edges, which are the spanner edges having endpoints in both sets. In other words, they are crossing the middle line $\mu(n)$. The values $i_L$ and $j_L$ are set according to \Cref{cor:sep} so that the spanner edges crossing $\mu(n)$ cannot be used for the left set; otherwise the resulting stretch will be bigger than $1+\eps$. For example, the path formed by the dashed red edges have total length of $16$ which is more than $3/2$ (i.e. more than $1+\eps$) the distance between the two points they connect (which is $6$). 
Similarly, $i_R$ and $j_R$ are such that the edges considered for the right set do not cross $\mu(u)$. In conclusion, the edges for the left set, right set and the cross edges will be mutually disjoint.
Thus, the recurrence for the lower bound is $T_2(n)= 2T_2(\floor{n/2}) + \Omega(n)$, which resolves to $T_2(n) = \Omega(n\log{n})$. 
}
\label{fig:lb2}
\end{figure}
\begin{lemma}\label{lemma:2prim}
For any positive integer $n$, and any two integers $l,r$ such that $r = l + n-1$, let $U(l,r)$ be a uniform line metric with $n$ points and let $X$ be a subspace of $U(l,r)$ which contains at least $\frac{31}{32}n$ points. Then, for any choice of  $\eps \in [0,1/2]$, any   spanner on $X$ with hop-diameter 2 and stretch $1+\eps$ contains at least $T'_2(n) \ge 0.48 \cdot \frac{n}{16}\log{n}$ edges which have both endpoints inside $[l,r]$.
\end{lemma}
\begin{proof}
Recall the recurrence used in \Cref{lemma:lb2}, $T_2(n) =2 T_2(\floor{n/2}) + \frac{11n}{64}$, which provides a lower bound on the number of edges of any $(1+\eps)$-spanner with hop-diameter 2 for $U(l,r)$ for any $l$ and $r$ such that $r = l+n-1$. The base case for this recurrence occurs when the considered interval contains less than 128 points. Consider the recursion tree of $T_2(n)$ and denote its depth by $\ell$ and the number of nodes at depth $i$ by $c_i$. In addition, denote by $n_{i,j}$ the number of points in the $j$th interval of the $i$th recursion level and by $e_{i,j}$ the number of cross edges contributed by this interval. By definition, we have
\begin{align*}
T_2(n) 
&= \sum_{i=0}^{\ell}\sum_{j=1}^{c_i} e_{i,j} 
\ge \frac{n}{16}\cdot  \log{n}.
\end{align*}

Recursion tree of $T_2(n)$ contains one node at depth $i=0$ which corresponds to an interval of $n$ points, two nodes at depth $i=1$ corresponding to intervals of $\floor{n/2}$ points, etc. Letting $n_i$ denote the total size of intervals corresponding to nodes at depth $i$ of recursion tree, we get
\begin{align*}
n_i = 2^i\cdot\floor*{\frac{n}{2^i}} \ge n-2^i + 1.
\end{align*}
The base case of the recurrence occurs whenever the considered interval contains less than 128 points. In other words, leaves of the recursion tree have $n_{i,j} < 128$.
Since the size of all the intervals at depth $i$ is the same  and it equals to $\floor{n/2^{i}}$, it follows that the leaves of the tree are at some level $i$ which satisfies $64 \le \floor{n/2^{i}} < 127$. At that level, $n_i \ge 0.98\cdot n$. For some level $i$ of the recursion tree, we refer to all the $n-n_i$ points which are not contained in any of the intervals of that level as \emph{ignored points}. Denote the collection of intervals containing ignored points by \emph{ignored intervals}.

Let $H'$ be any $(1+\eps)$-spanner on $X$ with hop-diameter 2. To lower bound the number of spanner edges in $H'$, we now consider the same recursion tree but take into consideration the fact that we are working on metric $X$, which is a subspace of $U(l,r)$. Hence, at each level of recursion, instead of $n$ points, there are at least $31n/32$ points in all the intervals of that level. We call the $j$th interval in the $i$th level \emph{good} if it contains at least $15n_{i,j}/16$ points from $X$. (Recall that we have used $n_{i,j}$ to denote the number of points from $U(l,r)$ in the $j$th interval of the $i$th level.) From the definition of good interval and the fact the intervals of each recursion level contain together at least $31n/32$ points, it follows that there are at least $n/2$ points contained in the good intervals at the $i$th level. If more than $n/2$ points were contained in bad intervals, and any such interval had at least $1/16$-fraction of points missing, it means that we would have more than $n/32$ points missing, which contradicts that we have at least $31n/32$ points in $X$. Recalling that we have at most $0.02n$ ignored points, we conclude that there is at least $0.48n$ points contained in the good intervals which were not ignored.

In \Cref{lemma:lb2}, we worked on the uniform line metric $U(1,n)$ and considered two intervals $L$ and $R$. We lower bounded the number of cross edges in this interval by $\min(|L|,|R|)\ge n/4-2 \ge 11n/64$, which holds for all $n \ge 26$. Observe that if the interval $[1,n]$ contained at least $15n/16$ points, rather than $n$, the same bound on the number of cross edges would hold, since $\min(|L|,|R|)-n/16 \ge n/4-2-n/16\ge 11n/64$, for all $n \ge 128$. This means that every good interval of size $n_{i,j}$ contributes at least $11n_{i,j}/64$ cross edges, as long as it is not considered as the base case (i.e. as long as $n_{i,j} \ge 128$). The same reasoning is applied inductively, so it holds also for $n_{i,j}$, rather than $n$, for any $i,j$.

Next, we analyze the recurrence $T'_2(n)$ representing the contribution of the intervals to the number of spanner edges for $X$. A lower bound on $T'_2(n)$ will provide a lower bound on the number of spanner edges of any spanner $H'$ of $X$ as defined above. Denote by $\Gamma_i$ all the good intervals in the $i$th level which were not ignored and by $e'_{i,j}$ the number of cross edges contributed by the $j$th interval in the $i$th level. 
Then we have
\begin{align*} 
T'_2(n) 
&= \sum_{i=0}^{\ell}\sum_{j=1}^{c_i} e'_{i,j} \ge \sum_{i=0}^{\ell}\sum_{j\in \Gamma_i} e_{i,j} \ge 0.48\cdot T_2(n) \ge 0.48 \cdot \frac{n}{16}\log{n}
\end{align*}
as claimed. 
This completes the proof of \Cref{lemma:2prim}.
\end{proof}
\begin{proof}[Completing the proof of \Cref{thm:2}] Note that $\alpha_2(n)=\ceil{\log{n}}$ and hence, we will show that $T'_2(n)\geq \frac{n}{256}\ceil{\log n}$. Suppose without loss of generality that we are working on any $t$-sparse line metric with $n$ points, $U((1,N),t)$, where $N = nt$. Let $H$ be an arbitrary  $(1+\eps)$-spanner for $U((1,N),t)$ with $((1,N),t)$-global hop-diameter 2. We would like to lower bound the number of $((l,r),t)$-global edges required for $H$. Let $M=\floor{n/2}t$ and let $L$ be the set of $((l,r),t)$-intervals that are fully inside $[N/4,M]$ and $R$ be the set of $((l,r),t)$-intervals that are fully inside $[M,3N/4]$. In that case, the number of $((l,r),t)$-intervals inside $L$ can be lower bounded by $|L|\ge\floor{(M-N/4+1)/t}\ge n/4-2$, which is the bound that we used for $L$. Similarly, we obtain that $|R|\ge n/4-1$. The cross edges will be those edges that contain one endpoint in $[1,M]$ and the other endpoint in $[M+1,N]$. It follows that the cross edges are also $((l,r),t)$-global edges. The same argument can be applied to lower bound the number of cross edges, implying the lower bound on the number of $((l,r),t)$-global edges. The same proof in \Cref{lemma:2prim} gives:
	\begin{equation*}
		T'_2(n)\geq 0.48\cdot \frac{n}{16}\log{n}\geq \frac{n}{256}\ceil{\log n}, 
	\end{equation*}
when $n\geq 1000$, as desired. 
\end{proof}

\subsection{Hop diameter 3}\label{sec:lb3}
\begin{theorem}\label{thm:3}
For any two positive integers $n\ge 1000$ and $t$, and any two integers $l,r$ such that $r = l + nt-1$, let $U((l,r),t)$ be a $t$-sparse line metric with $n$ points and let $X$ be a subspace of $U((l,r),t)$ which contains at least $\frac{127}{128}n$ points. Then, for any choice of $\eps \in [0,1/2]$, any spanner on $X$ with $((l,r),t)$-global hop-diameter 3 and stretch $1+\eps$ contains at least $T'_3(n) \ge \frac{n}{1024}\cdot\alpha_{3}(n)$ $((l,r),t)$-global edges which have both endpoints inside $[l,r]$.
\end{theorem}
\begin{remark}
Recall that we consider Steiner spanners, which could possibly contain additional Steiner points from the uniform line metric. 
\end{remark}
\begin{remark}
\Cref{thm:3} can be extended to $\eps > 1/2$. The only required change in the proof is to decrease the lengths of intervals by a factor of $1+\eps$, as provided by \Cref{lemma:sep}; it is readily verified that, as a result, the lower bound decreases by a factor of $\Theta(\eps)$.
\end{remark}

The theorem is proved in three steps. First, we prove \Cref{lemma:lb3}, which concerns uniform line metrics. Then, we prove \Cref{lemma:3prim} for a subspace that contains at least 31/32 fraction of the points of the original metric. In the third step, we observe that the same argument applies for $t$-sparse line metrics.
\begin{lemma}\label{lemma:lb3}
For any positive integer $n$, and any two integers $l,r$ such that $r = l + n-1$, let $U(l,r)$ be a uniform line metric with $n$ points. Then, for any choice of  $\eps \in [0,1/2]$, any spanner on $U(l,r)$ with hop-diameter 3 and stretch $1+\eps$ contains at least $T_3(n) \ge \frac{n}{40} \log\log{n}$ edges which have both endpoints inside $[l,r]$.
\end{lemma}

\begin{proof}
Suppose without loss of generality that we are working on the uniform line metric $U(1,n)$. Let $H$ be an arbitrary $(1+\eps)$-spanner for $U(1,n)$ with hop-diameter 3.

For the base case, we assume that $n < 128$. We have that $ \frac{n}{40} \log\log{n} < n-1$, which is a trivial lower bound on the number of edges of $H$.

We now assume that $n \ge 128$. Divide the the interval $[1,n]$ into consecutive subintervals containing $b\coloneqq\floor{\sqrt{n}}$ points: $[1, b], [b+1, 2b]$, etc. 
Our goal is to show that the number of spanner edges is lower bounded by $T_3(n)$, which satisfies recurrence
\begin{align*}
T_3(n) = \floor*{\frac{n}{\floor*{\sqrt{n}}}} \cdot T_3\left(\floor*{\sqrt{n}}\right) +  \Omega(n),
\end{align*}
with the base case $T_3(n) = (n/40)\log\log{n}$ when $n< 128$.

For any $j$ such that $1 \le j \le \floor{n/b}$, the interval spanned by the $j$th subinterval is $[(j-1)b+1, jb]$. Using the induction hypothesis, any spanner on $U((j-1)b+1, jb)$ contains at least $T_3(b)$ edges that are inside $[(j-1)b+1, jb]$. This means that all the subintervals will contribute at least $\floor{n/b} \cdot T_3(b)$ spanner edges that are mutually disjoint and in addition do not go outside of $[1,n]$. We will show that there are  $\Omega(n)$ edges that have endpoints in two different subintervals, called \emph{cross edges}. By definition, the set of cross edges is disjoint from the set of spanner edges considered in the term $\floor{n/b}\cdot T_3(b)$.

Consider the points that are within interval $[n/4, 3n/4]$. From \Cref{cor:sep}, since $n$ is sufficiently large, we know that any $(1+\eps)$-spanner path connecting two points in $[n/4,3n/4]$ has to have all its points inside $[1,n]$. 

We call a point \emph{global} if it is adjacent to at least one cross edge. Otherwise, the point is \emph{non-global}. The following two claims bound the number of cross edges induced by global and non-global points, respectively.
\begin{claim}\label{claim:g3}
Suppose that among points inside interval $[n/4,3n/4]$,
$m$ of them are global. Then, they induce at least $m/2$ spanner edges.
\end{claim}
\begin{proof}
Each global point contributes at least one cross edge and each edge is counted at most twice.
\end{proof}

\begin{claim}\label{claim:ng3}
Suppose that among points inside interval $[n/4,3n/4]$, $m$ of them are non-global. Then, they induce at least $\binom{m/\sqrt{n}}{2}$ cross edges. 
\end{claim}
\begin{proof}
Consider two sets $A$ and $B$ such that $A$ contains a non-global point $a \in [n/4,3n/4]$ and $B$ contains a non-global point $b \in [n/4,3n/4]$. Since $a$ is non-global, it can be connected via an edge either to a point inside of $A$ or to a point outside of $[1,n]$. Similarly, $b$ can be connected to either a point inside of $B$ or to a point outside of $[1,n]$. From \Cref{cor:sep}, and since $a,b \in [n/4,3n/4]$, we know that every spanner path with stretch $(1+\eps)$ connecting $a$ and $b$ has to use points inside $[1,n]$. This means that the spanner path with stretch $(1+\eps)$ has to have a form $(a,a',b',b)$, where $a' \in A$ and $b' \in B$.
In other words, we have to connect points $a'$ and $b'$ using a cross edge; furthermore every pair of intervals containing at least one non-global point induce one such edge and for every pair this edge is different. 

Each interval contains at most $b = \floor{\sqrt{n}}$ non-global points, so the number of sets containing at least one non-global point is at least $m/b$.
Interconnecting all the sets requires $\binom{m/b}{2} \ge \binom{m/\sqrt{n}}{2}$ edges.
\end{proof}

The number of points inside $[n/4,3n/4]$ is at least $n/2+1$, but we shall use a slightly weaker lower bound of $15n/32$.
We consider two complementary cases. In the first case, at least $1/4$ of $15n/32$ points are global. \Cref{claim:g3} implies that the number of the cross edges induced by these points is at least $15n/256$.
The other case is that at least $3/4$ fraction of $15n/32$ points are non-global. \Cref{claim:ng3} implies that for a sufficiently large $n$, the number of cross edges induced by these points can be lower bounded by $15n/256$ as well.
In other words, we have shown that in both cases, the number of cross edges is at least $\frac{15}{256}n > \frac{n}{18}$. Thus, we have:
\begin{align*}
T_3(n) &\geq \floor*{\frac{n}{\floor*{\sqrt{n}}}}  \cdot T_3\left(\floor{\sqrt{n}}\right) + \frac{n}{18} ~\stackrel{\mbox{\tiny{induction}}}{\ge}~ \floor{\sqrt{n}} \cdot \frac{\floor{\sqrt{n}}}{40}(\log\log\floor{\sqrt{n}}) + \frac{n}{18}
~\ge~ \frac{n}{40}\log\log{n}~,
\end{align*}
as claimed.
\end{proof}

\begin{lemma}\label{lemma:3prim}
For any positive integer $n$, and any two integers $l,r$ such that $r = l + n-1$, let $U(l,r)$ be a uniform line metric with $n$ points and let $X$ be a subspace of $U(l,r)$ which contains at least $\frac{127}{128}n$ points. Then, for any choice of $\eps \in [0,1/2]$, any spanner on $X$ with hop-diameter 3 and stretch $1+\eps$ contains at least $T'_3(n) \ge 0.18\cdot\frac{n}{40}\log\log{n}$ edges which have both endpoints inside $[l,r]$.
\end{lemma}
\begin{proof}
Recall the recurrence used in \Cref{lemma:lb3}, $T_3(n) = \floor{n/\floor{\sqrt{n}}}\cdot T_3(\floor{\sqrt{n}}) + \frac{n}{18}$, which provides a lower bound on the number of edges of any $(1+\eps)$-spanner with hop-diameter 3 for $U(l,r)$ for any $l$ and $r$ such that $r = l+n-1$. The base case for this recurrence occurs whenever the considered interval contains less than 128 points. Consider the recursion tree of $T_3(n)$ and denote its depth by $\ell$ and the number of nodes at depth $i$ by $c_i$. In addition, denote by $n_{i,j}$ the number of points in the $j$th interval of the $i$th recursion level and by $e_{i,j}$ the number of cross edges contributed by this interval. By definition, we have

\begin{align*}
T_3(n) 
&= \sum_{i=0}^{\ell}\sum_{j=1}^{c_i} e_{i,j}
~\ge~ \frac{n}{40}\cdot \log\log{n}
\end{align*}

Recursion tree of $T_3(n)$ contains one node at depth $i=0$ which corresponds to an interval of $n$ points, $\floor{n/\floor{\sqrt{n}}}$ nodes at depth $i=1$ corresponding to intervals of $\floor{\sqrt{n}}$ points, etc. Letting $n_i$ denote the total size of intervals corresponding to nodes at depth $i$ of the recursion tree, we get
\begin{align*}
n_i \ge n - \sum_{j=1}^{i-1} n ^ {1 - 1/2^j}.
\end{align*}
Since the size of all the intervals at depth $i$ is the same  and it equals to $\floor{n^{1-1/2^i}}$, it follows that the leaves of the tree are at some level $i$ which satisfies $11 \le \floor{n^{1-1/2^i}} < 127$. At that level, $n_i \ge 0.68\cdot n$. For some level $i$ of the recursion tree, we refer to all the $n-n_i$ points which are not contained in any of the intervals of that level as \emph{ignored points}. Denote the collection of intervals containing ignored points by \emph{ignored intervals}.

Let $H'$ be any $(1+\eps)$-spanner on $X$ with hop-diameter 3. To lower bound the number of spanner edges in $H'$, we now consider the same recursion tree but take into consideration the fact that we are working on metric $X$, which is a subspace of $U(l,r)$. Hence, at each level of recursion,
instead of $n$ points, there are at least $127n/128$ points in all the intervals of that level. We call the $j$th interval in the $i$th level \emph{good} if it contains at least $63n_{i,j}/64$ points from $X$. (Recall that we have used $n_{i,j}$ to denote the number of points from $U(l,r)$ in the $j$th interval of the $i$th level.) From the definition of a good interval and the fact the intervals of each recursion level contain together at least $127n/128$ points, it follows that there are at least $n/2$ points contained in the good intervals at the $i$th level. If more than $n/2$ points were contained in bad intervals, and any such interval had at least $1/64$-fraction of points missing, it means that we would have more than $n/128$ points missing; this contradicts that we have at least $127n/128$ points in $X$. In conclusion, at each level $i$, we have at least $0.18n$ points contained in the good intervals which were not ignored.

In \Cref{lemma:lb3}, we worked on the uniform line metric $U(1,n)$ and considered the points inside $[n/4,3n/4]$. We lower bounded the number of points inside $[n/4,3n/4]$ by $15n/32$. Observe that if the interval $[1,n]$ contained at least $63n/64$ points, rather than $n$, the same bound on the number of cross edges would hold, for any $n\ge 128$. This means that every good interval of size $n_{i,j}$ contributes at least $n_{i,j}/18$ cross edges. The same reasoning is applied inductively, so it holds also for $n_{i,j}$ rather than $n$ for any $i,j$.

Next, we analyze the recurrence $T'_3(n)$ representing the contribution of the intervals to the number of spanner edges for $X$. A lower bound on $T'_3(n)$ will provide a lower bound on the number of edges of any spanner $H'$ of $X$ as defined above.
Denote by $\Gamma_i$ all the good intervals in the $i$th level which were not ignored and by $e'_{i,j}$ the number of cross edges contributed by the $j$th interval in the $i$th level. Then, we have

\begin{align*}
T'_3(n) 
&= \sum_{i=0}^{\ell}\sum_{j=1}^{c_i} e'_{i,j}~\ge~ \sum_{i=0}^{\ell}\sum_{j\in \Gamma_i} e_{i,j}~=~ 0.18\cdot T_3(n) \ge 0.18\cdot\frac{n}{40}\log\log{n}
\end{align*}
as claimed. This completes the proof of \Cref{lemma:3prim}.
\end{proof}

\begin{proof}[Completing the proof of \Cref{thm:3}] Note that $\alpha_3(n)=\ceil{\log\log{n}}$ and hence, we will show that $T'_3(n) \ge \frac{n}{1024}\cdot \ceil{\log\log n}$. Suppose without loss of generality that we are working on any $t$-sparse line metric with $n$ points, $U((1,N),t)$, where $N = nt$. Let $H$ be an arbitrary  $(1+\eps)$-spanner for $U((1,N),t)$ with $((1,N),t)$-global hop-diameter 3. We would like to lower bound the number of $((l,r),t)$-global edges required for $H$. Let consider the set of $((l,r),t)$-intervals that are fully inside $[N/4,3N/4]$. The number of such intervals can be lower bounded by $((3N/4-N/4)/t-2\ge n/2-2$, which is larger than the bound of $15n/32$, which we used.
The cross edges will become $((1,N),t)$-global edges and the same argument can be applied to lower bound their number.  The same proof in \Cref{lemma:3prim} gives:
\begin{equation*}
	T'_3(n)\geq 0.18\cdot \frac{n}{40}\log\log{n}\geq   \frac{n}{1024}\cdot \ceil{\log\log n}
\end{equation*}
when $n\geq 1000$, as desired. 
\end{proof}

\section{Lower bound for constant hop-diameter}\label{sec:lbk}

We proceed to prove our main result, which is a generalization of \Cref{thm:main}. In particular, invoking \Cref{thm:k} stated below where $X$ is the uniform line metric $U(1,n)$ gives \Cref{thm:main}.

\begin{theorem}\label{thm:k}
For any two positive integers $n\ge 1000$ and $t$, and any two integers $l,r$ such that $r = l + nt-1$, let $U((l,r),t)$ be a $t$-sparse line metric with $n$ points and let $X$ be a subspace of $U((l,r),t)$ which contains at least $n(1-\frac{1}{2^{k+4}})$ points. Then, for any choice of $\eps \in [0,1/2]$ and any integer $k\ge 2$, any spanner on $X$ with $((l,r),t)$-global hop-diameter $k$ and stretch $1+\eps$ contains at least $T'_k(n) \ge \frac{n}{2^{6\floor{k/2}+4}} \cdot \alpha_{k}(n)$ $((l,r),t)$-global edges which have both endpoints inside $[l,r]$.
\end{theorem}
\begin{remark}
Recall that we consider Steiner spanners, which could possibly contain additional Steiner points from the uniform line metric. 
\end{remark}
\begin{remark}
\Cref{thm:k} can be extended to $\eps > 1/2$. The only required change in the proof is to decrease the lengths of intervals by a factor of $1+\eps$, as provided by \Cref{lemma:sep}; it is readily verified that, as a result, the lower bound decreases by a factor of $\Theta(\eps)$.
\end{remark}
\begin{proof}
We will prove the theorem by double induction on $k\ge 2$ and $n$. 
The base case for $k=2$ and $k=3$ and every $n$ is proved in \Cref{thm:2,thm:3}, respectively. 

For every $k\ge 4$, we shall prove the following two assertions.
\begin{assertions}
\item\label{assertion:1} For any two positive integers $n$ and $t$, and any two integers $l,r$ such that $r = l + nt-1$, let $U((l,r),t)$ be a $t$-sparse line metric with $n$ points. Then, for any choice of $ \eps \in [0,1/2]$, any spanner on $U((l,r),t)$ with $((l,r),t)$-global hop-diameter $k$ and stretch $1+\eps$ contains at least $T_k(n) \ge \frac{n}{2^{6\floor{k/2}+2}} \alpha_{k}(n)$ $((l,r),t)$-global edges which have both endpoints inside $[l,r]$. 

\item\label{assertion:2} For any two positive integers $n$ and $t$, and any two integers $l,r$ such that $r = l + nt-1$, let $U((l,r),t)$ be a $t$-sparse line metric with $n$ points and let $X$ be a subspace of $U((l,r),t)$ which contains at least $n(1-\frac{1}{2^{k+4}})$ points. Then, for any choice of $\eps \in [0,1/2]$, any spanner on $X$ with $((l,r),t)$-global hop-diameter $k$ and stretch $1+\eps$ contains at least $T'_k(n) \ge \frac{n}{2^{6\floor{k/2}+4}} \cdot \alpha_{k}(n)$ $((l,r),t)$-global edges which have both endpoints inside $[l,r]$.
\end{assertions}

For every $k \ge 4$, we first prove the first assertion, which relies on the second assertion for $k-2$. Then, we prove the second assertion which relies on the first assertion for $k$. 
We proceed to prove \cref{assertion:1}.

\paragraph*{Proof of \cref{assertion:1}.}
Suppose without loss of generality that we are working on any $t$-sparse line metric $U((1,N),t)$. Let $H$ be an arbitrary $(1+\eps)$-spanner for $U((1,N),t)$ with $((1,N),t)$-global hop-diameter $k$.

For the base case, we take $1 \le n \le 10000$. We have $\frac{n}{2^{6\floor{k/2}}+2} \alpha_{k}(n)$, which is at most $\frac{n}{2^{6\floor{k/2}}+2}\log^*(n)\le n-1$, a trivial lower bound on the number of the edges of any spanner.

Next, we prove the induction step. We shall assume the correctness of the two statements: (i) for $k$ and all smaller values of $n$,
and (ii) for $k' < k$ and all values of $n$. 
Let $N \coloneqq nt$ and let $b\coloneqq\alpha_{k-2}(n)$. Divide the the interval $[1,N]$ into consecutive $((1,N),bt)$-intervals containing $b$ points: $[1, bt], [bt+1, 2bt]$, etc. We would like to prove that the number of spanner edges is lower bounded by recurrence
\begin{align*}
T_k(n) = \floor*{\frac{n}{\alpha_{k-2}(n)}}\cdot T_{k}(\alpha_{k-2}(n)) + \Omega{\left(\frac{n}{2^{3k}}\right)},
\end{align*}
with the base case $T_k(n) = \frac{n}{2^{6\floor{k/2}+2}}\alpha_k(n)$ for $n \le 10000$.

There are $\floor{n/b}$ $((1,N),bt)$-intervals containing exactly $b$ points. For any $j$ such that $1 \le j \le \floor{n/b}$, the $j$th $((1,N),bt)$-interval is $[(j-1)bt+1, jbt]$. Using inductively the \cref{assertion:1} for $k$ and a value $b < n$, any spanner on $U((j-1)bt+1, jbt)$ contains at least $T_k(b)$ edges that are inside $[(j-1)bt+1, jbt]$. This means that all the $((1,N),bt)$-intervals will contribute at least $\floor{n/b} \cdot T_k(b)$ spanner edges that are mutually disjoint and in addition do not go outside of $[1,N]$.

We will show that there are $\Omega(n/2^{3k})$  edges that have endpoints in two different $((1,N),bt)$-intervals, i.e. edges that are $((1,N),bt)$-global. Since these edges are $((1,N),bt)$-global, they are disjoint from the spanner edges considered in the term $\floor{n/b}\cdot T_3(b)$. We shall focus on points that are inside $((1,N),bt)$-intervals fully inside $[N/4, 3N/4]$; 
denote the number of such points by $p$. We have $p\ge n/2-2\alpha_{k-2}(n)$, but we will use a   weaker bound:
\begin{equation} \label{e:pnt}
p\ge n/4. 
\end{equation}

\begin{definition}
A point that is incident on at least one $((1,N),bt)$-global edge is called a $((1,N),bt)$-global point.
\end{definition}

Among the $p$ points inside inside $[N/4, 3N/4]$, denote by $p'$ the number of $((1,N),bt)$-global points. Let $p'' = p - p'$, and $m$ be the number of $((1,N),bt)$-global edges incident on the $p$ points. 
Since each $((1,N),bt)$-global point contributes at least one $((1,N),bt)$-global edge and each such edge is counted at most twice,
we have 
\begin{equation}\label{e:global}
m ~\ge~ p' / 2.
\end{equation}

Next, we prove that 
\begin{equation} \label{e:main}
m ~\ge~ \frac{n}{2^{6\floor{k/2}+1}}, \mbox{~~if~} \ceil*{\frac{p''}{b}} 
~\ge~ \left(1-\frac{1}{2^{k+2}}\right)\cdot\ceil*{\frac{p}{b}}
\end{equation}
Recall that we have divided $[1,N]$ into consecutive $((1,N),bt)$-intervals containing $b\coloneqq\alpha_{k-2}(n)$ points.
Consider now all the $((1,N),bt)$-intervals that are fully inside $[N/4,3N/4]$, and denote this collection of $((1,N),bt)$-intervals by $\mathcal{C}$. Let $l'$ (resp. $r'$) be the leftmost (resp. rightmost) point 
of the leftmost (resp. rightmost) interval in $\mathcal{C}$; note that $l'$ and $r'$ may not coincide with points of the input metric,
they are simply the leftmost and rightmost boundaries of the intervals in $\mathcal{C}$. 

\paragraph*{Constructing a new line metric.}
For each $((1,N),bt)$-interval $I$ in $\mathcal{C}$, if $I$ contains a point that is not $((1,N),bt)$-global, assign an arbitrary such point in $I$ as its representative; otherwise, assign an arbitrary point as its representative. The collection  $\mathcal{C}$ of $((1,N),bt)$-intervals, together with the set of representatives uniquely defines $(bt)$-sparse line metric, $U((l',r'),bt)$.  
This metric has $\ceil{p/b}$ $((1,N),bt)$-intervals, since there are $\ceil{p/b}$ intervals covering $p$ points in the input $t$-sparse metric $U((1,N),t)$ inside the interval $[N/4,3N/4]$. Recall from \Cref{def:tsparse} that a $bt$-sparse metric is uniquely defined given its $((1,N),bt)$-intervals and representatives. 
Let $X$ be the subspace of $U((l',r'),bt)$ induced by the representatives of all intervals in $\mathcal{C}$ that contain points that are not $((1,N),bt)$-global and using \Cref{e:main},
we have 
\begin{equation} \label{eq:fraction}
|X| ~\ge~ \ceil*{\frac{p''}{b}} 
~\ge~ \left(1-\frac{1}{2^{k+2}}\right)\cdot\ceil*{\frac{p}{b}}
\end{equation}

Recall that $H$ is an arbitrary $(1+\eps)$-spanner for $U((1,N),t)$ with $((1,N),t)$-global hop-diameter $k$.
Let $a$ and $b$ be two arbitrary points in $X$, and denote their corresponding $((1,N),bt)$-intervals by $A$ and $B$, respectively.
Since $a$ (reps., $b$) is not $((1,N),bt)$-global, it can be adjacent either to points outside of $[1,N]$ or to points inside $A$ (resp., $B$). 
By \Cref{cor:sep} and since $a,b\in[N/4,3N/4]$, any spanner path with stretch $(1+\eps)$ connecting $a$ and $b$ must remain inside $[1,N]$. Hence, any $(1+\eps)$-spanner path in $H$ between $a$ and $b$ is of the form $(a, a',\dots, b', b)$, where $a' \in A$ (resp. $b' \in B$). Consider now the same path in the metric $X$. It has at most $k$ hops, where the first and the last edges are not $((1,N),bt)$-global. Thus, although this path contains at most $k$ $((1,N),t)$-global edges in $U((1,N),t)$, it has at most $k-2$ $((1,N),bt)$-global edges in $X$. 
It follows that $H$ is a (Steiner) $(1+\eps)$-spanner with $((1,N),bt)$-global hop-diameter $k-2$ for $X$. See \Cref{fig:lb4} for an illustration.

Denote by $n' \coloneqq \ceil{p/b}$ the number of points in $U((l',r'), bt)$.
Since $p\ge n/4$, it follows that $n' \ge \ceil{n/(4b)}$. 
By \labelcref{eq:fraction}, $X$ is a subspace of $U((l',r'), bt)$, 
and its size is at least a $(1-1/2^{k+2})$-fraction (i.e., a $(1-1/2^{(k-2)+4})$-fraction) of that of $U((l',r'), bt)$.
Hence, by the induction hypothesis of \cref{assertion:2} for $k-2$,
 we know that any spanner on $X$ with $((l',r'),bt)$-global hop-diameter $k-2$ and stretch $1+\eps$ contains at least $T'_{k-2}(n') \ge \frac{n'}{2^{6\floor{(k-2)/2}+4}} \cdot \alpha_{k-2}(n')$ $((l',r'),bt)$-global edges which have both endpoints inside $[l',r']$.
Since every $((l',r'),bt)$-global edge is also a $((1,N),bt)$-global edge, 
we conclude with the following lower bound on the number of $((1,N), bt)$-global edges required by $H$:

\begin{align*}
T'_{k-2}\left(n'\right) &\ge \frac{n'}{2^{6\floor{(k-2)/2}+4}}\cdot\alpha_{k-2}\left(n'\right)\\
&\ge \frac{n}{4\cdot 2^{6\floor{(k-2)/2}+4}\cdot\alpha_{k-2}(n)}\cdot\alpha_{k-2}{\left(\ceil*{\frac{n}{4\alpha_{k-2}(n)}}\right)}\\
&\ge \frac{n}{8\cdot 2^{6\floor{(k-2)/2}+4}}\\
&= \frac{n}{2^{6\floor{k/2}+1}}
\end{align*}
The last inequality follows since, when $k\ge 4$, the ratio between $\alpha_{k-2}(\ceil{n/4\alpha_{k-2}(n)})$ and $\alpha_{k-2}(n)$ can be bounded by $1/2$ for sufficiently large $n$ (i.e. larger than the value considered in the base case).
In other word, we have shown that whenever $\ceil{{p''}/{b}} 
~\ge~ (1-{1}/{2^{k+2}})\cdot\ceil{{p}/{b}}$, the number of the $((1,N),bt)$-global edges incident on the $p$ points inside $[N/4,3N/4]$ is lower bounded by ${n}/{2^{6\floor{k/2}+1}}$; we have thus proved \labelcref{e:main}.

Recall (see \labelcref{e:pnt}) that we lower bounded the number $p$ of points inside $[N/4, 3N/4]$ as $p\ge n/4$. We consider two complementary cases: either $\ceil{p''/b} ~\ge~ (1-{1}/{2^{k+2}})\cdot\ceil{p/b}$, or $\ceil{p''/b} ~<~ (1-{1}/{2^{k+2}})\cdot\ceil{p/b}$, where $p''$ is the number of points in $[N/4,3N/4]$ that are not $((1,N),bt)$-global. In the former case (i.e. when $\ceil{p''/b} ~\ge~ (1-{1}/{2^{k+2}})$), 
by \labelcref{e:main}, we have the number of $((1,N),bt)$-global edges is lower bounded by $n/2^{6\floor{k/2}+1}$. In the latter case,  we have
\begin{align*}
\frac{p-p'}{b} - 1 < \floor*{\frac{p-p'}{b}} = \frac{p''}{b} < \left(1-\frac{1}{2^{k+2}}\right)\cdot\ceil*{\frac{p}{b}}  < \left(1-\frac{1}{2^{k+2}}\right)\cdot\frac{p}{b} + 1.
\end{align*}
In other words, we can lower bound $p'$ by $p/2^{k+2} - 2b$. From \labelcref{e:global} and using that $p \ge n/4$,
the number of $((1,N),bt)$-global edges is lower bounded by $n/2^{k+5}-\alpha_{k-2}(n)$. Since the former bound is always smaller for $n$ sufficiently large (i.e. larger than the value considered in the base case), we shall use it as a lower bound on the number of $((1,N),bt)$-global edges required by $H$. We note that every $((1,N),bt)$-global edge is also $((1,N),t)$-global, as required by \cref{assertion:1}. It follows that
\begin{align*}
T_k(n) &\ge \floor*{\frac{n}{\alpha_{k-2}(n)}}\cdot\frac{\alpha_{k-2}(n)}{2^{6\floor{k/2}+2}}\cdot\alpha_k(\alpha_{k-2}(n)) + \frac{n}{2^{6\floor{k/2}+1}}\\
&\ge \left(\frac{n}{\alpha_{k-2}(n)}-1 \right)\cdot \frac{\alpha_{k-2}(n)}{2^{6\floor{k/2}+2}} \cdot (\alpha_k(n)-1) + \frac{n}{2^{6\floor{k/2}+1}}\\
&\ge \frac{n}{2^{6\floor{k/2}+2}}\cdot\alpha_k(n) - \frac{2n}{2^{6\floor{k/2}+2}}  + \frac{n}{2^{6\floor{k/2}+1}}\\
&= \frac{n}{2^{6\floor{k/2}+2}} \alpha_k(n)
\end{align*}
For the second inequality we have used \Cref{lemma:alphakstep}, and for the third, the fact that $\alpha_{k-2}(n) \cdot (\alpha_{k}(n) - 1) \le n$ for sufficiently large $n$ (i.e. larger than the value considered in the base case).
This concludes the proof of \cref{assertion:1}. 

\begin{figure}[!bht]
\centering
\input{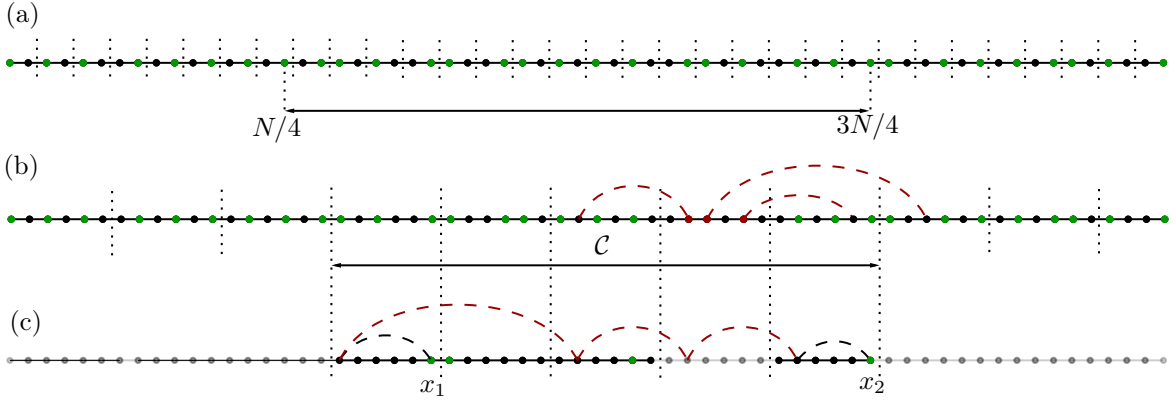}
\caption{Constructing a new line metric and invoking the induction hypothesis. \textbf{(a)} We have $n=32$, $k=5$, and a $2$-sparse line metric $U((1,64), 2)$ with representatives of each $((1,64), 2)$-interval highlighted in green. \textbf{(b)} Since $b = \alpha_{k-2}(n) = 3$, we consider a collection of $((1,64), 6)$-global intervals inside $[N/4, 3N/4]$, denoted by $\mathcal{C}$. The seventh block contains only $((1,64), 6)$-global points (highlighted in red) as each of them is incident on a $((1,64), 6)$-global edge. \textbf{(c)} The new line metric is $6$-sparse line metric $U((19, 48), 6)$ consisting of 4 green points. Finally, we use the induction hypothesis of \cref{assertion:2} for $k=3$ to lower bound the number of $((1,N), 6)$-global edges. A spanner path between $x_1$ and $x_2$ consisting of 5 edges, 3 of which are  $((1,N), 6)$ global is depicted. }
\label{fig:lb4}
\end{figure}
\paragraph*{Proof of \cref{assertion:2}.~}

Suppose without loss of generality that we are working on any $t$-sparse line metric $U((1,N),t)$. Let $H$ be an arbitrary $(1+\eps)$-spanner for $U(1,N)$ with $((1,N),t)$-global hop-diameter $k$.

We shall inductively assume the correctness of \cref{assertion:1} and \cref{assertion:2}: (i) for $k$ and all smaller values of $n$,
and (ii) for $k' < k$ and all values of $n$.

Recall the recurrence we used in the proof of \cref{assertion:1}, $T_k(n) = \floor{n/\alpha_{k-2}(n)}\cdot T_k(\alpha_{k-2}(n)) + \Omega(n/2^{3k})$, which provides a lower bound on the number of $((l,r),t)$-global edges of $H$.
The base case for this recurrence is whenever $n<10000$. Consider the recursion tree of $T_k(n)$ and denote its depth by $\ell$ and the number of nodes at depth $i$ by $c_i$. In addition, denote by $n_{i,j}$ the number of points in the $j$th interval of the $i$th level and by $e_{i,j}$ the number of $((1,N),bt)$-global edges contributed by this interval. We have that the contribution of an interval is $n_{i,j}/2^{6\floor{k/2}+1}$. By definition, we have 
\begin{align*}
T_k(n) 
&= \sum_{i=1}^{\ell}\sum_{j=1}^{c_i} e_{i,j} 
\ge \frac{n}{2^{6\floor{k/2}+2}}\alpha_k(n)
\end{align*}

Let $H'$ be any $(1+\eps)$ spanner on $X$ with $((1,N),t)$-global hop-diameter $k$. To lower bound the number of spanner edges in $H'$, we now consider the same recursion tree, but take into consideration the fact that we are working on metric $X$, which is a subspace of $U((1,N),t)$. 
This means that at each level of recursion, instead of $n$ points, there is at least $n(1-1/2^{k+4})$ points in $X$. 
The contribution of the $j$th interval in the $i$th level is denoted by $e'_{i,j}$. 
We call the $j$th interval in the $i$th level \emph{good} if it contains at least $n_{i,j}(1-1/2^{k+3})$ points from $X$. (Recall that we have used $n_{i,j}$ to denote the number of points from $U(l,r)$ in the $j$th interval of the $i$th level.) From the definition of good interval and the fact that each level of recurrence contains at least $n(1-1/2^{k+4})$ points, it follows that there are at least $n/2$ points contained in the good intervals at the $i$th level. Denote the collection of all the good intervals at the $i$th level by $\Gamma_i$.

Recall that we are working with recurrence $T_k(n) = \floor{n/\alpha_{k-2}(n)}\cdot T_k(\alpha_{k-2}(n)) + \Omega(n/2^{3k})$. In particular, in the first level of recurrence, we consider the contribution of $n$ points, whereas in the second level, we consider the contribution of $\floor{n/\alpha_{k-2}(n)}\cdot \alpha_{k-2}(n)$ points. Denote by $n_i$ the number of points whose contribution we consider in the $i$th level of recurrence. Then, we have $n_1 = n$, $n_2 = \floor{n/\alpha_{k-2}(n)}\cdot \alpha_{k-2}(n) \ge n - \alpha_{k-2}(n)$. Denote by $\alpha_{k-2}^{(j)}(n)$ value of $\alpha_{k-2}(\cdot)$ iterated on $n$, i.e. $\alpha_{k-2}^{(0)}(n) = n$, $\alpha_{k-2}^{(1)}(n) = \alpha_{k-2}(n)$, $\alpha_{k-2}^{(2)}(n) = \alpha_{k-2}(\alpha_{k-2}(n))$, etc. In general, for $i\ge 2$, we have
\begin{align*}
n_i &\ge n- \sum_{j=2}^{i}\frac{n\alpha_{k-2}^{(j-1)}(n)}{\alpha_{k-2}^{(j-2)}(n)}\\
&\ge n - n \cdot\sum_{j=2}^{i}\frac{\ceil*{\log^{(j-1)}(n)}}{\ceil*{\log^{(j-2)}(n)}}.
\end{align*}
We observe that there is an exponential decay between the numerator and denominator of terms in each summand and that terms grow with $j$. Since we do not consider intervals in the base case, we also know that $\ceil{\log^{(i-1)}(n)} \ge 10000$, meaning that the largest term in the sum is $10000 / 2^{9999}$. By observing that every two consecutive terms increase by a factor larger than 2, we conclude that $n_i \ge 0.99n$.
Since at each level there are at least $n/2$ points inside of good intervals, this means that there are at least  $0.49n$ points inside of good intervals which were not ignored. Denote by $\Gamma_i$ the set of good intervals in the $i$th level whose contribution is not ignored. Then we have
\begin{align*}
T'_k(n) 
&= \sum_{i=1}^{\ell}\sum_{j=1}^{c'_i} e'_{i,j} 
~\ge~ \sum_{i=1}^{\ell}\sum_{j\in \Gamma_i} e_{i,j}
~\ge~ 0.49\cdot T_k(n) \ge \frac{n}{2^{6\floor{k/2}+4}}\alpha_k(n),
\end{align*}
as claimed. 
This concludes the proof of \cref{assertion:2}.

We have thus completed the inductive step for $k$.
\end{proof}

\bibliographystyle{alpha}

\bibliography{refs,spanner,latex8}
\appendix
\section{Tradeoff using two-parameter Ackermann function}\label{sec:tradeoff}
In this \namecref{sec:tradeoff}, we prove that the tradeoff of \cite{AS87} of $k$ vs. $\Omega(n \alpha_k(n))$ between hop-diameter and number of edges implies the tradeoff of \cite{Yao82} and \cite{CG06} of $\Omega(\alpha(m,n))$ vs. $m$. We also prove a variant of this result, 
relevant to our lower bound (\Cref{thm:main}) that has an exponential on $k$ slack in the number of edges.
Specifically, we prove \Cref{lemma:tradeoff} and \Cref{cor:tradeoff}.

We start with making the following simple claim.  
\begin{claim}\label{clm:ackermann}
For every $i \ge 0$ and every $j \ge 4$, we have $A(i+1, j) \ge A(i, 2j)$.
\end{claim}
\begin{proof}
Observe that $A(1,j) = 2^j$. We have $A(i+1, j) = A(i, A(i+1, j-1)) \ge A(i, 2^{j-1}) \ge A(i, 2j)$, whenever $j \ge 4$.
\end{proof}

\begin{lemma}\label{lemma:tradeoff}
Let $\mathcal{G}$ be a graph family and let $k'$, $n_0$ be a constant such that for every integer $n'\ge n_0$ any $k'$-hop spanner on an $n'$-vertex graph $G_{n'}$ from $\mathcal{G}$ has at least $n' \alpha_{k'}(n')$ edges. Then, for any choice of integers $n\ge n_0$ and $m\ge n$, any $m$-edge spanner on an $n$-vertex graph $G_{n}$ from $\mathcal{G}$ has hop diameter of at least $\alpha(m,n)$.
\end{lemma}

\begin{proof}
Let $H$ be an arbitrary $m$-edge spanner on $G_{n}$ from $\mathcal{G}$ as in the lemma statement.
If $n\alpha_0(n) \le m$, then $\alpha(m,n)$ is a small constant, and the hop-diameter of $H$ is trivially $\Omega(1)$. 
We henceforth assume that $n\alpha_0(n) > m$, hence there is a unique $k$ such that 
\begin{equation}
n\cdot\alpha_{k}(n) ~\le~ m ~<~  n\cdot \alpha_{k-1}(n).\label{eq:region}
\end{equation}
From the lemma statement, the hop-diameter of $H$ is greater than $k-1$, i.e., it is at least: 

\begin{align*}
k &\ge \min\left\{i \mid A(i, \min\left\{s \mid A(k, s) \ge n \right\}) \ge n \right\}\\
&\ge \min\left\{i \mid A(i, \alpha_{2k}(n)) \ge n \right\} &\text{ (by definition of $\alpha_{2k}(n)$)}\\
&\ge \min\left\{i \mid A(i, \alpha_{2k}(n)) > \log{n} \right\}\\
&\ge \min\left\{i \mid A(i, \alpha_{k}(n)) > \log{n} \right\}\\
&\ge \min\left\{i \mid A\left(i, \frac{m}{n}\right) > \log{n} \right\} &\text{ (\Cref{eq:region}) }\\
&\ge \min\left\{i \mid A\left(i, 4\ceil*{\frac{m}{n}}\right) > \log{n} \right\}\\
&= \alpha(m,n)
\end{align*}
This completes the proof of \Cref{lemma:tradeoff}.
\end{proof}

We proceed to proving \Cref{cor:tradeoff}, which is a variant of \Cref{lemma:tradeoff},  relevant to our lower bound (\Cref{thm:main}) that has an exponential on $k$ slack in the number of edges.
By \Cref{thm:main}, any $(1+\eps)$-spanner with hop-diameter $k$ on any $n$-point uniform line metric must have at least $\frac{n}{2^{6\floor{k/2}+4}} \alpha_{k}(n)$ edges. (This refined bound on the number of edges, which does not use the $\Omega$-notation, is due to \Cref{thm:k}, which is a generalization of \Cref{thm:main}.)

\begin{corollary}\label{cor:tradeoff}
For any two positive integers $m$ and $n$ such that $m < n^2/32$, let $k$ be the unique integer such that $\frac{n}{2^{6\floor{k/2}+4}}\cdot \alpha_{k}(n) \le m < \frac{n}{2^{6\floor{(k-1)/2}+4}}\cdot \alpha_{k-1}(n)$. For any choice of $\eps \in [0, 1/2]$, any $m$-edge $(1+\eps)$-spanner for the uniform line metric with $n$ points must have hop diameter of at least $\Omega(\alpha(m,n)-k)$.
\end{corollary}
\begin{proof}
The value $k$ as in the statement always exists since
for $k=0$, we have $\frac{n}{2^{6\floor{k/2}+4}} \alpha_{k}(n) \ge \frac{n^2}{32} > m$ and the sequence $\{\frac{n}{2^{6\floor{k/2}+4}} \alpha_{k}(n)\}_{k=0}^\infty$ is decreasing.

Similarly to the proof of \Cref{lemma:tradeoff}, we can lower bound $k$ as follows.
\begin{align*}
k &\ge \min\left\{i \mid A(i, \min\left\{s \mid A(k, s) \ge n \right\}) \ge n \right\}\\
&\ge \min\left\{i \mid A(i, \alpha_{2k}(n)) \ge n \right\} &\text{ (by definition of $\alpha_{2k}(n)$)}\\
&\ge \min\left\{i \mid A(i, \alpha_{2k}(n)) > \log{n} \right\}\\
&\ge \min\left\{i \mid A(i, \alpha_{k}(n)) > \log{n} \right\}\\
&\ge \min\left\{i \mid A\left(i, 2^{6\floor{k/2}+4}\cdot\frac{m}{n}\right) > \log{n}\right\}\\
&\ge\min\left\{i \mid A\left(i + 6\floor{k/2}+4,\frac{m}{n}\right) > \log{n}\right\} &\text{(\Cref{clm:ackermann})}\\
&\ge\min\left\{i \mid A\left(i,\frac{m}{n}\right) > \log{n}\right\} - 6\floor{k/2}-4\\
&\ge\min\left\{i \mid A\left(i,4\ceil*{\frac{m}{n}}\right) > \log{n}\right\} - 6\floor{k/2}-4\\
&\ge\alpha(m,n) - 6\floor{k/2}-4\\
&= \Omega(\alpha(m,n) - k)
\end{align*}
This completes the proof of \Cref{cor:tradeoff}.
\end{proof}

\end{document}